\documentclass[10pt,journal,twocolumn,romanappendices]{IEEEtran}
\usepackage{cite,calc}
\usepackage[cmex10]{amsmath}
\usepackage{amsfonts,psfrag,amssymb,graphicx,subfigure,amsthm}
\usepackage[mathcal]{euscript}  %

\usepackage{amsmath,amssymb,array,graphicx,epsfig,dsfont,enumerate,url,textcomp}
\usepackage{times,epsf}
\begin{document}
\newcommand{\HH}[1]{\mathcal{H}_b\left({#1}\right)}
\newcommand{\uzer}{\underline{0}}
\newcommand{\uV}{\underline{V}}
\newcommand{\uA}{\underline{A}}
\newcommand{\uD}{\underline{D}}
\newcommand{\uv}{\underline{v}}
\newcommand{\uT}{\underline{T}}
\newcommand{\ut}{\underline{t}}
\newcommand{\ur}{\underline{r}}
\newcommand{\uR}{\underline{R}}
\newcommand{\uc}{\underline{c}}
\newcommand{\uC}{\underline{C}}
\newcommand{\ul}{\underline{l}}
\newcommand{\uL}{\underline{L}}
\newcommand{\uh}{\underline{h}}
\newcommand{\uH}{\underline{H}}
\newcommand{\ue}{\underline{e}}
\newcommand{\uE}{\underline{E}}
\newcommand{\uG}{\underline{G}}
\newcommand{\ug}{\underline{g}}
\newcommand{\uz}{\underline{z}}
\newcommand{\uZ}{\underline{Z}}
\newcommand{\uu}{\underline{u}}
\newcommand{\uU}{\underline{U}}
\newcommand{\uj}{\underline{j}}
\newcommand{\uJ}{\underline{J}}
\newcommand{\uX}{\underline{X}}
\newcommand{\ux}{\underline{x}}
\newcommand{\uY}{\underline{Y}}
\newcommand{\uy}{\underline{y}}
\newcommand{\uW}{\underline{W}}
\newcommand{\uw}{\underline{w}}
\newcommand{\uth}{\underline{\theta}}
\newcommand{\uTh}{\underline{\Theta}}
\newcommand{\uph}{\underline{\phi}}
\newcommand{\ual}{\underline{\alpha}}
\newcommand{\uxi}{\underline{\xi}}
\newcommand{\us}{\underline{s}}
\newcommand{\uS}{\underline{S}}
\newcommand{\un}{\underline{n}}
\newcommand{\uN}{\underline{N}}
\newcommand{\up}{\underline{p}}
\newcommand{\uq}{\underline{q}}
\newcommand{\uf}{\underline{f}}
\newcommand{\ua}{\underline{a}}
\newcommand{\ub}{\underline{b}}
\newcommand{\uDelta}{\underline{\Delta}}
%  CALLIGRAPHIC
\newcommand{\cA}{{\cal A}}
\newcommand{\cB}{{\cal B}}
\newcommand{\cC}{{\cal C}}
\newcommand{\cc}{{\cal c}}
\newcommand{\cD}{{\cal D}}
\newcommand{\cE}{{\cal E}}
\newcommand{\cI}{{\cal I}}
\newcommand{\cK}{{\cal K}}
\newcommand{\cL}{{\cal L}}
\newcommand{\cN}{{\cal N}}
\newcommand{\cP}{{\cal P}}
\newcommand{\cQ}{{\cal Q}}
\newcommand{\cR}{{\cal R}}
\newcommand{\cS}{{\cal S}}
\newcommand{\cs}{{\cal s}}
\newcommand{\cT}{{\cal T}}
\newcommand{\ct}{{\cal t}}
\newcommand{\cU}{{\cal U}}
\newcommand{\cV}{{\cal V}}
\newcommand{\cW}{{\cal W}}
\newcommand{\cX}{{\cal X}}
\newcommand{\cx}{{\cal x}}
\newcommand{\cY}{{\cal Y}}
\newcommand{\cy}{{\cal y}}
\newcommand{\cZ}{{\cal Z}}
%  TILDE
\newcommand{\tE}{\tilde{E}}
\newcommand{\tZ}{\tilde{Z}}
\newcommand{\tz}{\tilde{z}}
%   HAT
\newcommand{\hU}{\hat{U}}
\newcommand{\hX}{\hat{X}}
\newcommand{\hY}{\hat{Y}}
\newcommand{\hZ}{\hat{Z}}
\newcommand{\huX}{\hat{\uX}}
\newcommand{\huY}{\hat{\uY}}
\newcommand{\huZ}{\hat{\uZ}}
\newcommand{\indp}{\underline{\; \| \;}}
\newcommand{\diag}{\mbox{diag}}
\newcommand{\sumk}{\sum_{k=1}^{K}}
\newcommand{\beq}[1]{\begin{equation}\label{#1}}
\newcommand{\eeq}{\end{equation}}
%%% FROM SHLOMO SHAMAI

\newcommand{\bg}{\mbox{\boldmath \begin{math}g\end{math}}}
\newcommand{\ba}{{\bf a}}
\newcommand{\bb}{{\bf b}}
\newcommand{\bc}{{\bf c}}
\newcommand{\bD}{{\bf D}}
\newcommand{\bbf}{{\bf f}}
\newcommand{\bn}{{\bf n}}
\newcommand{\bs}{{\bf s}}
\newcommand{\bt}{{\bf t}}
\newcommand{\bu}{{\bf u}}
\newcommand{\bv}{{\bf v}}
\newcommand{\bx}{{\bf x}}
\newcommand{\by}{{\bf y}}
\newcommand{\bz}{{\bf z}}
\newcommand{\bC}{{\bf C}}
\newcommand{\bJ}{{\bf J}}
\newcommand{\bN}{{\bf N}}
\newcommand{\bQ}{{\bf Q}}
\newcommand{\bS}{{\bf S}}
\newcommand{\bT}{{\bf T}}
\newcommand{\bV}{{\bf V}}
\newcommand{\bX}{{\bf X}}
\newcommand{\bY}{{\bf Y}}
\newcommand{\bZ}{{\bf Z}}
\newcommand{\oI}{\overline{I}}
\newcommand{\oD}{\overline{D}}
\newcommand{\oh}{\overline{h}}
\newcommand{\oV}{\overline{V}}
\newcommand{\oR}{\overline{R}}
\newcommand{\oH}{\overline{H}}
\newcommand{\ol}{\overline{l}}
\newcommand{\E}{{\cal E}_d}
\newcommand{\el}{\ell}
\newcommand{\modulo}{\: \mathrm{mod} \:}
%%%%%%%%%%%%%%%%%% KESAL

%%%%%%%%%%%%%%%%%%%%%%%%%%%%%%%%%%%%%%%

\newcommand{\InProd}[2]{\langle {#1}\ \mathbf{,} \ 
{#2}\rangle}                     %% Inner Product
\newcommand{\IntLim}[4]{\int_{#1}^{#2} {#3}\, 
\mathrm{d}{#4}}                        %% Integral with limits
%\newcommand{\EXP}[1]{\exp\left({#1}\right)}                                    
%     %% Exponential with parantheses
\newcommand{\EXP}[1]{e^{{#1}}}                                         %% 
%%Exponential with parantheses
\newcommand{\LOG}[1]{\log\left({#1}\right)}                                     
    %% Logarithm with parantheses
\newcommand{\Expect}[1]{\mathrm{E}\left[{#1}\right]}                            
    %% Expected value
\newcommand{\CONV}[2]{\left({#1}\circledast{#2}\right)}                         
    %% Circular Convolution
\newcommand{\NORM}[1]{\left\|{#1}\right\|}                                      
    %% Norm operator
\newcommand{\NORMC}[1]{\left|{#1}\right|}
\newcommand{\SNR}{\mathrm{SNR}}                                                 
    %% SNR font
\newcommand{\PDF}{f}                                                            
 %% Generic pdf
\newcommand{\PDFZ}{\PDF_Z}                                                      
       %% Generic pdf
\newcommand{\PDFtZ}{\PDF_t^{(Z)}}
\newcommand{\DEC}{{Y}}
\newcommand{\CDF}{F}                                                            
 %% Generic cdf
\newcommand{\EE}{d}
\newcommand{\CD}{\PDF'}                                                         
      %% Inverse error variable
\newcommand{\CDG}{\CD^{(\sigma_Z^2)}}                                           
    %% Channel derivative variable generalized
\newcommand{\CDZ}{\CD_Z}
\newcommand{\CN}{\PDF_{\mathrm{n}}}                                             
    %% Channel noise variable
\newcommand{\CNG}{\CN^{(\sigma_Z^2)}}                                           
    %% Channel noise density generalized
\newcommand{\ID}{{\mathcal{I}}_{\Delta}}                                        
    %% Symmetric interval
\newcommand{\DD}{[-\frac{\Delta}{2},\frac{\Delta}{2})}                          
    %% Symmetric interval
%\newcommand{\ID}{{\DD}}                                                        
%          %% Symmetric interval
%\newcommand{\TTT}{{\mathcal{T}}_{P}^{(\Delta)}}                            %% 
%%%Feasible set of transfers
\newcommand{\TTT}{{\mathcal{G}}(\Delta,\SNR)}                            %% 
%%Feasible set of transfers
\newcommand{\SETEE}{{\mathcal{D}}_{\SNR}^{(\Delta)}}
\newcommand{\MOD}[1]{\left\lceil 
{#1}\right\rfloor_{\Delta}}                        %% Modulo-Delta operator
\newcommand{\ENCO}{\Gamma_{t}^{(\textsc{E})}}                                   
    %% Encoder generic
\newcommand{\RPLUS}{\mathbb{R}^{+}_0}                                           
 %% Non-negative reals
\newcommand{\EXPXK}[2]{\EXP{-\frac{\left({#1}-{#2}\Delta\right)^2}{2\sigma_Z^2}}}
\newcommand{\TINVD}{g^{-1}}
\newcommand{\TINV}[1]{\TINVD\!\!\left(#1\right)}
\newcommand{\PRB}[1]{\mathrm{Pr}\left(#1\right)}
\newcommand{\DEL}{\frac{\Delta}{2}}
\newcommand{\DERIV}[1]{\frac{\mathrm{d}}{\mathrm{d}{#1}}}
\newcommand{\PDERIV}[1]{\frac{\partial}{\partial{#1}}}
\newcommand{\SUMINF}[2]{\sum_{#1=-\infty}^{\infty}{#2}}
\newcommand{\CAP}{\mathrm{C}}
\newcommand{\MUTUAL}[2]{\mathrm{I}\left({#1};{#2}\right)}
\newcommand{\ENTROPY}[1]{\mathrm{h}\left({#1}\right)}
\newcommand{\PROB}[1]{\mathrm{Pr}\left({#1}\right)}
\newcommand{\INTD}[2]{\int_{\ID} {#1}\, \mathrm{d}{#2}} %% Integral with 
%%Delta/2 limits

\newcommand{\FT}[1]{\widetilde{#1}}
\newcommand{\FC}[1]{\widetilde{#1}}
\newcommand{\LEB}{{\mathcal{L}}^2(\ID)}
\newcommand{\EXPK}[1]{e^{ \jmath 2\pi {#1} \frac{x}{\Delta}}}
\newcommand{\EQUP}[1]{\stackrel{\text{(#1)}}{=}}
%%%%%%%%%%%%%%%%%%%%%%%%%%%%%%%%%%%%%%

\newtheoremstyle{normal}% hnamei
{3pt}% hSpace abovei
{3pt}% hSpace belowi
{}% hBody fonti
{}% hIndent amounti1
{\itshape}% hTheorem head fonti
{:}% hPunctuation after theorem headi
{.5em}% hSpace after theorem headi2
{}% h
\theoremstyle{normal}
\newtheorem{proposition}{Proposition}
\newtheorem*{theorem_b}{Theorem~1$^\prime$}
\newtheorem{fact}{Fact}
\newtheorem{lemma}{Lemma}
\newtheorem{definition}{Definition}
\newtheorem{corollary}{Corollary}
\newtheorem{theorem}{Theorem}
\newtheorem{remark}{Remark}
\newtheorem{note}{Note}
\newtheorem{assumption}{Assumption}
\newtheorem{notation}{Notation}
\newtheorem{example}{Example}

\newenvironment{correct}%
{\noindent\ignorespaces}%
{\par\noindent%
\ignorespacesafterend}
%%%%%%%%%%%%%%%%%%%%%%%%%%%%%%%%%%%%%%%%%%%%%%%KESALEND

\title{Generalized Bounds on the Capacity of the Binary-Input Channels}

\author{Mustafa~Kesal} %, and Uri~Erez,~\IEEEmembership{Member,~IEEE}
%\thanks{This work was supported in part by the Israel Science Foundation, 
%%%grant 
%ISF1234/09 and  by the Binational Science Foundation, grant
%BSF2008455.    This work was presented in part at the
%IEEE International Symposium on Information Theory, Austin, TX,
%June 13-18, 2010.}
%\thanks{M.~Kesal and U.~Erez are with the Department of
%Electrical Engineering - Systems, Tel Aviv University, Ramat Aviv,
%69978, Israel (Email: \{kesal,uri\}@eng.tau.ac.il).}
%}

\author{\normalsize Mustafa Kesal \\
%\small EECS, MIT \\[-5pt]
%\small Texas A\&M University \\[-5pt]
%\small College Station, TX 77840-3128 \\[-5pt]
\small mkesal@aselsan.com.tr
%\and
%\normalsize Uri Erez \\
%\small Dept.  EE - Systems \\ Tel Aviv University \\
%\small uri@eng.tau.ac.il
}
%
%\author{
%\authorblockN{Mustafa Kesal}
%\authorblockA{
%%EE - Systems Dprt. \\
%%Tel Aviv University\\
%%Tel Aviv, Israel \\
%Email: mkesal@gmail.com }
%\and
%\authorblockN{Uri Erez}
%\authorblockA{Dept. EE - Systems  \\
%Tel Aviv University\\
%Tel Aviv, Israel \\
%Email: uri@eng.tau.ac.il }
%}
% avoiding spaces at the end of the author lines is not a problem with
% conference papers because we don't use \thanks or \IEEEmembership
%Homer Simpson\authorrefmark{2},
%Montgomery Scott\authorrefmark{3} and
%Eldon Tyrell\authorrefmark{4}}
%\authorblockA{\authorrefmark{3}Starfleet Academy, San Francisco, California 
%96678-2391\\
%Telephone: (800) 555--1212, Fax: (888) 555--1212}
%\authorblockA{\authorrefmark{4}Tyrell Inc., 123 Replicant Street, Los Angeles, 
%California 90210--4321}}

\maketitle

%%%%%%%%%%%%%%%%%%%%%%%%%%%%%%%%%%%%%%%%%%%%%%%%%%%%%%%%%%%%%%%%%%%%%%%
%%%%%%%%%%%%%%%%%%%%%%%%      Abstract      %%%%%%%%%%%%%%%%%%%%%%%%%%%
%%%%%%%%%%%%%%%%%%%%%%%%%%%%%%%%%%%%%%%%%%%%%%%%%%%%%%%%%%%%%%%%%%%%%%%

\begin{abstract}
For the class of the memoryless binary-input channels which are {\em{not}} 
necessarily 
symmetric, we derive tight bounds on the 
capacity in terms of the Bhattacharyya parameter. As it turns out, the 
bounds derived 
under the {\em{symmetric}} channel assumption in \cite{fabregas} are valid for 
the general case as well.
\end{abstract}

{\em {\textbf{KeyWords:} {Polar Codes, symmetric capacity, binary-input 
channels, Bhattacharyya 
parameter.}}}

%%%%%%%%%%%%%%%%%%%%%%%%%%%%%%%%%%%%%%%%%%%%%%%%%%%%%%%%%%%%%%%%%%%%%%%
%%%%%%%%%%%%%%%%%%%%%%%%      Introduction  %%%%%%%%%%%%%%%%%%%%%%%%%%%
%%%%%%%%%%%%%%%%%%%%%%%%%%%%%%%%%%%%%%%%%%%%%%%%%%%%%%%%%%%%%%%%%%%%%%%
\section{Introduction}
In \cite{arikan}, Arikan introduced the revolutionary concept of channel 
polarization. As 
a tool to prove the capacity achieving property of the polar codes, the bounds 
(\ref{ineq:BoundArikan}) which relates the 
{\em{symmetric capacity}} to the {\em{Bhattacharyya parameter}} are 
utilized. In \cite{fabregas}, these bounds are improved under the memoryless 
{\em{symmetric}} binary-input channels context.  Specifically, the bounds 
presented in \cite{arikan} 
\begin{align}{\label{ineq:BoundArikan}}
\log_2\left(\frac{2}{1+Z(W)}\right)\leq I(W)\leq\sqrt{1-Z^2(W)}
\end{align}
are improved in \cite{fabregas} to
\begin{align}
(1-Z(W))\leq I(W)\leq 
1-\mathcal{H}_b\left(\frac{1-\sqrt{1-Z^2(W)}}{2}\right)
\end{align}
The bounds  are compared in Figures {\ref{Fig:ImprovementLowerBound}} and 
{\ref{Fig:ImprovementUpperBound}}.

{\begin{figure}[ht]
  \begin{center}
  \footnotesize
  \includegraphics[scale=.5]{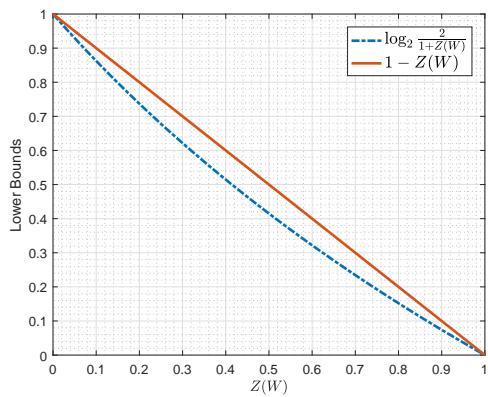}
   \caption{Lower bound provided in \cite{arikan} and the improved version.}
     \label{Fig:ImprovementLowerBound}
  \end{center}
\end{figure}}

{\begin{figure}[ht]
  \begin{center}
  \footnotesize
  \includegraphics[scale=.5]{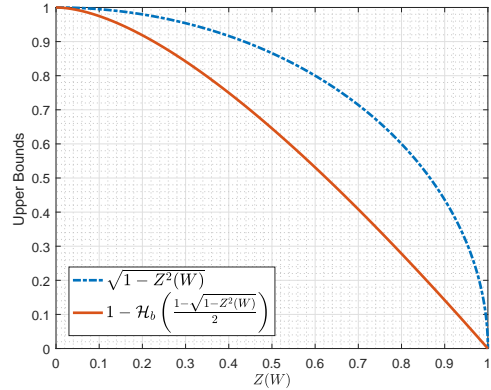}
   \caption{Upper bound provided in \cite{arikan} and the improved version.}
     \label{Fig:ImprovementUpperBound}
  \end{center}
\end{figure}}

In this paper, we derive capacity bounds for the binary-input memoryless 
channels that 
are not necessarily symmetric. As it turns out, the capacity of an arbitrary 
binary-input memoryless channel $W$, which is achieved by a 
non-uniform input distribution in general, lies within the bounds presented in 
\cite{fabregas} as well. 

Since the central idea utilized in \cite{fabregas} 
technique hinges on the separability of the symmetric channels, our method of 
proof considerably differs from the mentioned work. For the symmetric channels, 
as the uniform input 
density achieves the capacity, the symmetric capacity is {\em{indeed}} the 
channel capacity in the usual sense. Since  we do not 
assume any symmetric structure on the channel in our treatment, the symmetric 
capacity is not the right metric to work on. Moreover, the derivations in 
\cite{fabregas} are 
built around the idea that every binary-input memoryless {\em{symmetric}} 
channel admits a decomposition into subchannels that are binary symmetric 
channels. Once the symmetric channel assumption is abandoned, the essential 
ingredient of the proof (the subchannel 
decomposition) is not available anymore. In other words, although the results 
of \cite{fabregas} extends to the {\em{asymmetric}} channel models, the method 
does not.

The paper is organized as follows. In Section~\ref{sec:Notation} we present the 
notation followed in the paper along with the definitions 
of the parameters of interest. The lower bound derivation is presented in 
Section~\ref{sec:LowBound}, and we conclude our treatment with the upper bound 
result developed in Section~\ref{sec:UpBound}.

\section{Notation}{\label{sec:Notation}}

We consider binary input $N$-ary output discrete  channels, which we denote 
as $W$. The transition 
probabilities are $N$-ary vectors $\mathbf{P}$ and $\mathbf{Q}$ whose 
components are defined as
\begin{align}
p_n&=W(Y=n\mid X=0){\label{eq:Pn}}\\
q_n&=W(Y=n\mid X=1){\label{eq:Qn}}.
\end{align}
Since they are probability mass functions, they satisfy $\sum_{n} 
p_n=\sum_{n}q_n=1$. In  The source is 
assumed to 
produce bits independently and identically with $\mathrm{Pr}(X_n=0)=\alpha$, 
for all $n\in\mathbb{N}$.

Since our treatment involves differential calculus tools, 
we prefer to use natural logarithms, i.e., of base $e$ unless stated otherwise. 
We state the main results in {\em{bits}} whereas we switch to nats during 
the 
derivations. Therefore, for example, the binary entropy term appearing in 
(\ref{eq:MainLowerEquation}) is in terms of bits, although the 
definition (\ref{eq:BinEnt}) is given in nats. 

The {\em{binary entropy function}} and the {\em{binary Bhattacharyya 
parameter}} are defined as
\begin{align}
\mathcal{H}_b(p)&\triangleq -p\log(p)-(1-p)\log(1-p){\label{eq:BinEnt}}\\
\mathcal{Z}_b(p)&\triangleq \sqrt{4p(1-p)},
\end{align}
respectively.

The objective of this paper is to derive strong relationships between two 
parameters of interest, namely the {\em{channel capacity}} and the 
{\em{Bhattcharyya parameter}}. 

\begin{definition}[Channel Capacity] The {\em{capacity}} of binary-input 
memoryless channel $W$ is defined as
\begin{align*}
C(W)\triangleq \max_{\alpha\in[0,1]} I(W;\alpha),
\end{align*}
where
\begin{align*}
	I(W;\alpha)\triangleq \mathrm{H}(X)-\mathrm{H}(X\mid Y).
			\end{align*}
	Here the input random variable $X$ is Bernoulli distributed with 
	$\mathrm{Pr}(X=0)=\alpha$, and $\mathrm{H}(\cdot)$ is the entropy 
	functional.
	\end{definition}
Simple algebraic manipulations yield the following explicit form 
\begin{align}{\label{eq:altform}}
I(W;\alpha)=\mathcal{H}_b(\alpha)-\mathrm{E}_{\mathbf{R}}\left[\mathcal{H}_b\left(\frac{\alpha
 p_n}{r_n}\right)\right]
\end{align}
Here, the expectation is evaluated over the convex combination of the densities 
$\mathbf{R}\triangleq \alpha\mathbf{P}+(1-\alpha)\mathbf{Q}$.
Note that, for $\alpha=1/2$, the term $I(W;1/2)$ coincides with the {\em{ 
symmetric capacity}} 
as defined in ~\cite{arikan}. 
We conclude this section with the definition of a metric which is associated 
with the channel reliability.
\begin{definition}[Bhattacharyya Parameter] The {\em{Bhattacharyya parameter}} 
of the channel $W$ is defined as
	\begin{align*}
	Z(W)\triangleq \sum_{y_0=1}^N\sqrt{W(y_0\mid 0) W(y_0\mid 1)}.
	\end{align*}
\end{definition}

In the next section, we generalize the lower bound on the channel capacity as 
given in \cite{fabregas}.
%
%These parameters have specific forms for the {\em {binary output}} channels as 
%defined next.
%\begin{definition}[Binary Entropy Function] The {\em{binary entropy function}} 
%is defined as $\mathcal{H}_b(p)\triangleq -p\log p-(1-p)\log(1-p)$ for all 
%$p\in[0,1]$.
%\end{definition}
%and likewise
%\begin{definition}[Binary Bhattacharyya Parameter] The {\em{binary 
%Bhattacharyya 
%parameter}} 
%is defined as $\mathcal{Z}_b(p)\triangleq \sqrt{4p(1-p)}$ for all 
%$p\in[0,1]$.
%\end{definition}
%
%The symmetric capacity has the following equivalent expression which we prefer 
%to work on 
%\begin{align}{\label{eq:EquivExpression}}
%I(W)=\log 2-\sum_{n=1}^N 
%\frac{(p_n+q_n)}{2}\mathcal{H}_b\left(\frac{p_n}{p_n+q_n}\right).
%\end{align}

\section{Generalized Lower Bound}{\label{sec:LowBound}}
Our main approach is based on providing relationships between the binary 
parameters and then extending them to the $N$-point case. 
\begin{theorem}\label{thm:HUpperBound}
For any $p\in[0,1]$
\begin{align}{\label{eq:MainLowerEquation}}
\HH{p}\leq \mathcal{Z}_b(p)\quad [\mathrm{bits}].
\end{align}
\end{theorem}
\begin{proof}[Proof of Theorem \ref{thm:HUpperBound}] Rather than proving the 
statement directly, we prefer to work on  an equivalent problem and then reach 
the desired conclusion. Formally speaking, we claim, for any positive numbers 
$x$ 
and $y$ 
such that $x y\leq 1$, the following property holds  
\begin{align}{\label{eq:EquivStatement}}
(x+y)\HH{\frac{x}{x+y}}\leq (2\log 2)\sqrt{xy}.
\end{align}
The equivalence of (\ref{eq:MainLowerEquation}) and (\ref{eq:EquivStatement}) 
is revealed once the correspondence $p\to \frac{x}{x+y}$ is 
substituted in (\ref{eq:EquivStatement}).

As for the equivalent setting, let $\alpha_0\triangleq x y$ and without 
loss of 
generality, let $x\leq y$. Then, the task reduces to show that the maximum 
value of the function 
\begin{align*}
f(x)=\left(x+\frac{\alpha_0}{x}\right)\HH{\frac{x}{x+\frac{\alpha_0}{x}}}
\end{align*}
over the variable $x\in(0,\sqrt{\alpha_0}]$ is  $(2\log 
2)\sqrt{\alpha_0}$. We show that the 
function $f$ is monotone non-decreasing over the domain of interest and hence 
the maximum is attained at the end point $x^\star=\sqrt{\alpha_0}$. Since the 
function is continuous and well-behaving, we use the standard differential 
calculus approach, i.e., we show that the first order derivative of $f$ is 
non-negative over the interval of interest. The calculations are simple but to 
keep the analysis tractable, we  define the auxiliary parameters 
$u\triangleq x+ 
\frac{\alpha_0}{x}$ and $t\triangleq 
\frac{x}{u}=\frac{x^2}{x^2+\alpha_0}$. Both parameters are implicit functions 
of $x$ and the first order derivative $u^\prime$ is to be understood  as
$\frac{\mathrm{d} u}{\mathrm{d}x}$, which, incidentally, satisfies $u^\prime= 
(2t-1)/t.$
% Before we proceed, we 
%note the following  
%easy to verify fact
%\begin{align}
%u^\prime=\frac{2t-1}{t}.
%\end{align}
Since $f(x)=u\,\HH{\frac{x}{u}}$, the first order derivative is given by
\begin{align*}
f^\prime(x)=u^\prime \HH{\frac{x}{u}}+u \left(\frac{x}{u}\right)^\prime 
\log\left(\frac{u-x}{x}\right).
\end{align*}
Here, we used the chain rule for derivatives and the fact  that the first order 
derivative of the binary entropy function with
respect to its argument is given by $\HH{x}^\prime=\log(\frac{1-x}{x})$. 
Substituting the expressions given for the auxiliary parameters and applying  
basic algebraic manipulations yield \mbox{$f^\prime(x)=((1-t)\log(1-t)-t\log 
t)/{t}.$}
%\begin{align}
%f^\prime(x)=\left(\frac{2t-1}{t}\right)\HH{t}+(2-2t)\log\left(\frac{1-t}{t}\right),
%\end{align}
%which, after simple algebraic manipulations reduces to 
%$f^\prime(x)=((1-t)\log(1-t)-t\log t)/{t}.$
%\begin{align}
%f^\prime(x)=\frac{(1-t)\log(1-t)-t\log t}{t}.
%\end{align}
We need to show that the function $g(t)\triangleq (1-t)\log(1-t)-t\log t$ is 
positive over the interval $t\in(0,\frac{1}{2}]$. The restriction on $t$ is 
imposed by definition and it can be deduced from the monotone behavior of the 
function $h(x)=\frac{x^2}{x^2+\alpha_0}$. As for the function $g(t)$, we see 
that $g(0)=g(1/2)=0$. Therefore, it is sufficient to show that $g(t)$ is 
{\em{concave}} on $(0,1/2]$ and thus stays positive over the open interval 
$(0,1/2)$. To show that the function $g$ is indeed concave, we inspect the   
second order derivative $g^{\prime\prime}(t)=(2t-1)/(t-t^2)$.
Since $t\in(0,\frac{1}{2}]$, we see that $g^{\prime\prime}(t)<0$ over this 
region and hence $g(t)$ is a concave function with value $0$ at its end 
points, which proves the fact that $g(t)$ is non-negative over the interval 
$(0,1/2]$, so is $f^\prime(x)$. Finally, since $f(x)$ is an non-decreasing 
function over $x\in(0,\sqrt{\alpha_0}]$, we have $f(x)\leq 
f(\sqrt{\alpha_0})=(2\log 2)\sqrt{xy}$, as claimed.
\end{proof}

The lower bound on the channel capacity is stated as corallary to 
Theorem~\ref{thm:HUpperBound}.

\begin{corollary}{\label{corr:LowerBound}} The capacity of the channel $W$ is 
bounded below as
\begin{align}{\label{ineq:LowBd}}
C(W)\geq 1-Z(W)\quad [\mathrm{bits}].
\end{align}
\end{corollary}
\begin{proof}[Proof of Corollary \ref{corr:LowerBound}] By definition
\begin{align}
I(W;\alpha)&=\mathcal{H}_b(\alpha)-\mathrm{E}_{\mathbf{R}}\left[\mathcal{H}_b\left(\frac{\alpha
 p_n}{r_n}\right)\right]\nonumber\\
 &\geq \mathcal{H}_b(\alpha)-(\log 2)\mathrm{E}_{\mathbf{R}}\left[
 \mathcal{Z}_b\left(\frac{\alpha p_n}{r_n}\right)\right]\nonumber\\
 &=\mathcal{H}_b({\alpha})-(\log 2)\sum_{n=1}^N \sqrt{4\alpha (1-\alpha) p_n 
 q_n}\nonumber\\
 &=\mathcal{H}_b({\alpha})-(\log 2)\sqrt{4\alpha(1-\alpha)} 
 Z(W).{\label{eq:lastline}}
\end{align}
By definition, $C(W)=\max_{\alpha} I(W;\alpha)\geq I(W;1/2)$. Therefore, 
\begin{align}
C(W)\geq \mathcal{H}_b(1/2)-(\log 2) Z(W),
\end{align}
which is equivalent to (\ref{ineq:LowBd}) when converted to bits from nats.
\end{proof}

The upper bound is more involved and derived in the next section.
%In the proof, we judiciously picked $\alpha=1/2$ so as to obtain the bound. A 
%plausible question is if there is a better candidate. The answer turns out to 
%be negative. To begin with, due to the symmetry, we consider 
%$\alpha\in[0,1/2]$. The first order derivative of  the right hand side of 
%%%(\ref{eq:lastline}) is given by $\

\section{Generalized Upper Bound}{\label{sec:UpBound}}

In this section, we derive an improved upper bound on the capacity of 
an arbitrary binary input memoryless channel. We first provide a lemma that is 
central to our arguments.

\begin{lemma}{\label{lem:Convexity}} The binary entropy is a {\em{convex}} 
function of the associated Bhattacharyya parameter. That is, there exists a 
convex bijection $F:[0,1]\to[0,\log2]$ such that 
$\mathcal{H}_b(p)=F\left(\mathcal{Z}_b(p)\right)$ for all $p\in[0,1]$ and it is 
defined as
\begin{align}{\label{eq:Relation}}
F(x)=\mathcal{H}_b\left(\frac{1-\sqrt{1-x^2}}{2}\right).
\end{align}
\end{lemma}
\begin{proof}[Proof of Lemma~\ref{lem:Convexity}] We first show the existence 
of such function, and then we prove its convexity. In order to avoid notational 
confusion, we define 
\begin{align}
f(p)&\triangleq -p\log p-(1-p)\log(1-p)\\
g(p)&\triangleq\sqrt{4p(1-p)}.
\end{align}
It is clear that $\mathcal{H}_b(p)=f(p)$ and $\mathcal{Z}_b(p)=g(p)$. Now, 
since both the binary entropy and the binary Bhattacharyya parameters are 
symmetric around $p=1/2$, it suffices to prove our claims for the interval 
$p\in[0,1/2]$. Therefore, solving $p$ in terms of $\mathcal{Z}_b(p)$, we get
\begin{align}
p=\frac{1-\sqrt{1-\mathcal{Z}_b(x)^2}}{2},
\end{align}
which is simply the explicit form of $p=g ^{-1}(\mathcal{Z}_b(p))$.
Finally, substituting $p$ in $\mathcal{H}_b(p)$ yields
\begin{align}
\mathcal{H}_b(p)=f\left(\frac{1-\sqrt{1-\mathcal{Z}_b^2(p)}}{2}\right),
\end{align}
In other words, $\mathcal{H}_b(p)=F(\mathcal{Z}_b(p))$ where
\begin{align}
F(x)=\mathcal{H}_b\left(\frac{1-\sqrt{1-x^2}}{2}\right).
\end{align}
The bijective nature of $F$ is due to the fact that both $f$ and $g$ are both 
monotone increasing bijections over the interval of interest which implies the 
bijectivity of the function $(f\o g^{-1})$. This completes the proof of the 
first part.

As for the convex nature of $F$, we prefer to derive the property using the 
parametric form rather than via $F$ directly. Essentially, we need to show 
\begin{align}
\frac{\mathrm{d}^2\mathcal{H}_b(p)}{\mathrm{d}\mathcal{Z}_b(p)^2}>0,
\end{align}
for $p\in[0,1/2]$. Applying the chain rule twice yields
\begin{align}
\frac{\mathrm{d}^2\mathcal{H}_b(p)}{\mathrm{d}\mathcal{Z}_b^2(p)}
=\frac{f^{\prime\prime}(p)g^\prime(p)-f^\prime(p)g^{\prime\prime}(p)}{(g^{\prime}(p))^3},
\end{align}
where
\begin{align}
f^\prime(p)&=\log\frac{1-p}{p}\\
f^{\prime\prime}(p)&=\frac{-1}{p(1-p)}\\
g^\prime(p)&=\frac{1-2p}{\sqrt{p(1-p)}}\\
g^{\prime\prime}(p)&=\frac{-1}{2(p(1-p))^{3/2}}.
\end{align}
After simple algebraic manipulations, we get
\begin{align}
\frac{\mathrm{d}^2\mathcal{H}_b(p)}{\mathrm{d}\mathcal{Z}_b^2(p)}=\frac{\log\left(\frac{1-p}{p}\right)+4p-2}{2(1-2p)^3}.
\end{align}
Since the denominator is always positive for $p\in(0,1/2)$, we show the 
numerator, $t(p)\triangleq \log\left(\frac{1-p}{p}\right)+4p-2$, is positive as 
well. Since the first order derivative is 
\mbox{$t^\prime(p)=4-\frac{1}{p(1-p)}=\frac{-(2p-1)^2}{p(1-p)}$} is negative 
over the 
interval $(0,1/2)$, we deduce that the function  $t(p)$ is monotone decreasing, 
and hence $t(p)\geq t(1/2)=0$. All in all, we have
\begin{align}
\frac{\mathrm{d}^2\mathcal{H}_b(p)}{\mathrm{d}\mathcal{Z}_b^2(p)}=\frac{t(p)}{2(1-2p)^3}>0,
\end{align}
as claimed.
\end{proof}

We now provide the upper bound on the capacity for the binary input 
memoryless channels.

\begin{theorem}{\label{thm:UpperBound}} The  capacity of any 
memoryless binary input discrete channel has the following upper bound
\begin{align}
C(W)\leq 1-\mathcal{H}_b\left(\frac{1-\sqrt{1-Z^2(W)}}{2}\right)\quad 
[\mathrm{bits}].
\end{align}
\end{theorem}
\begin{proof}[Proof of Theorem~\ref{thm:UpperBound}]We substitute 
(\ref{eq:Relation}) in the alternative form as given 
in (\ref{eq:altform}) to get
\begin{align*}
I(W;\alpha)&=\mathcal{H}_b(\alpha)-\mathrm{E}_{\mathbf{R}}\left[\mathcal{H}_b\left(\frac{\alpha
 p_n}{r_n}\right)\right]\\
&=\mathcal{H}_b(\alpha)-\mathrm{E}_{\mathbf{R}}\left[F\left(\mathcal{Z}_b\left(\frac{\alpha
 p_n}{r_n}\right)\right)\right]\\
&\leq \mathcal{H}_b(\alpha)-
F\left(\mathrm{E}_{\mathbf{R}}\left[\mathcal{Z}_b\left(\frac{\alpha 
p_n}{r_n}\right)\right]\right)\\
&=\mathcal{H}_b(\alpha)-F\left(\sqrt{4\alpha(1-\alpha)}Z(W)\right)\\
&=f(\alpha;Z^2(W)),
\end{align*}
where 
\begin{align}
f(\alpha;\beta)\triangleq 
\mathcal{H}_b(\alpha)-\mathcal{H}_b\left(\frac{1-\sqrt{1-4\alpha(1-\alpha) 
\beta}}{2}\right)
\end{align}
We claim that given any $\beta\in[0,1]$, the function $f(\alpha;\beta)$ is 
maximized at $\alpha=1/2$, regardless of the value of the parameter $\beta$. To 
begin with, since 
\mbox{$f(\alpha;\beta)=f(1-\alpha;\beta)$,} for 
all $\alpha\in[0,1]$ and $\beta\in[0,1]$, it suffices to prove the claim for 
$\alpha\in[0,1/2]$. In other words, we show
\begin{align}
f(\alpha;\beta)\leq f(1/2;\beta),\quad \forall  \alpha\in[0,1/2]\ 
\mathrm{and}\  \forall
\beta\in[0,1]
\end{align}
We prove a stronger result which implies the above claim.
Specifically, we show $f(\alpha;\beta)$ is a monotone increasing function in 
$\alpha\in[0,1/2]$ for any fixed $\beta\in[0,1]$. To this end, we evaluate the 
first order derivative function of $f$, which is given by
\begin{align}{\label{eq:fDeriv}}
\frac{\partial f(\alpha;\beta)}{\partial \alpha}=g(\alpha;1)-g(\alpha;\beta),
\end{align}
where
\begin{align}
g(\alpha;\beta)\triangleq 
\frac{(1-2\alpha)\beta}{\sqrt{1-4\alpha(1-\alpha)\beta}}\log\left(\frac{1+\sqrt{1-4\alpha(1-\alpha)\beta}}{1-\sqrt{1-4\alpha(1-\alpha)\beta}}\right)
\end{align}
In order to prove the monotone increasing nature of $f(\alpha;\beta)$ over 
$\alpha$, we prove 
a sufficient property which states that that $g(\alpha;\beta)$ is a monotone 
increasing function over $\beta$, 
for any given $\alpha\in[0,1/2]$. Once again, we apply the first order 
derivative test to prove that claim. In order to keep the derivations 
tractable, we define
\begin{align}
t\triangleq \sqrt{1-4\alpha(1-\alpha)\beta},
\end{align}
which in turn gives $\beta=\frac{1-t^2}{4\alpha(1-\alpha)}$ with $t\in[0,1]$. 
Then, we 
alternatively have
\begin{align}
g(\alpha;\beta)=\frac{(1-2\alpha)}{4\alpha(1-\alpha)}\left(\frac{1-t^2}{t}\right)\log\left(\frac{1+t}{1-t}\right)
\end{align}
Using the chain rule and applying algebraic manipulations yield
\begin{align}
\frac{\partial g(\alpha;\beta)}{\partial \beta}&=\frac{\partial 
g(\alpha;\beta)}{\partial t}\frac{\partial t}{\partial \beta}\\
&=\frac{(1-2\alpha)(1+t^2)}{2t^3}\left[\log\left(\frac{1+t}{1-t}\right)-\frac{2t}{1+t^2}\right].
\end{align}
The first order derivative of the $h(t)\triangleq 
\log\left(\frac{1+t}{1-t}\right)-\frac{2t}{1+t^2}$ is given by 
$h^\prime(t)=\frac{8t^2}{(1-t^4)(1+t^2)}$ which is non-negative for all 
$t\in[0,1]$, Therefore, the function $h(t)$ is increasing over this interval 
which is equivalent to $h(t)\geq h(0)=0$. This proves that the partial 
derivative $g(\alpha;\beta)$ with respect to $\beta\in[0,1]$ is non-negative 
for any fixed $\alpha\in[0,1/2]$. In other words, $g(\alpha;1)\geq 
g(\alpha;\beta)$ over the admissible intervals defined for $\alpha$ and 
$\beta$. Therefore, by (\ref{eq:fDeriv}), we have $\frac{\partial 
f(\alpha;\beta)}{\partial \alpha}\geq 0$, which in turn implies 
$f(\alpha;\beta)\leq f(1/2;\beta)$. Finally, setting $\beta=Z^2(W)$ and 
converting the quantities from nats to bits, conclude the proof.
\end{proof}

\section{Conclusion}
In this paper, we  derive  bounds on the capacity of binary-input 
memoryless channels {\em{without}} the {symmetry} assumption. As it turns out, 
the bounds derived under the symmetric channel assumption in \cite{fabregas} 
apply to the general case while their methods are  applicable only to the
symmetric case.

\end{document}